\documentclass[11pt, onecolumn]{IEEEtran}

\usepackage{amsmath}
\usepackage{amssymb}
\usepackage{verbatim}
\usepackage{epsfig}

\newtheorem{lemma}{Lemma}

\newtheorem{theorem}{Theorem}

\newcommand{\bh}{{\bf h}}

\newcommand{\bw}{{\bf w}}

\newcommand{\bH}{{\bf H}}
\newcommand{\bG}{{\bf G}}

\newcommand{\bx}{{\bf x}}

\newcommand{\bv}{{\bf v}}

\newcommand{\bheff}{{\bf h}^{\textrm{eff}}}
\newcommand{\yeff}{y^{\textrm{eff}}}

\psfull

\begin{document}

\title{Antenna Combining for the MIMO Downlink Channel}
\author{\authorblockN{Nihar Jindal} \\
\authorblockA{University of Minnesota, Department of ECE\\
Minneapolis, MN 55455, USA \\
Email: nihar@umn.edu}}
\maketitle

\begin{abstract}
A multiple antenna downlink channel where limited channel
feedback is available to the transmitter is considered.  In
a vector downlink channel (single antenna
at each receiver), the transmit antenna array can be used to transmit separate
data streams to multiple receivers only if the transmitter has very
accurate channel knowledge, i.e., if there is high-rate channel
feedback from each receiver. In this work it is shown that
channel feedback requirements can be significantly reduced if
each receiver has a small number of antennas and
appropriately combines its antenna
outputs. A combining method that minimizes channel
quantization error at each receiver, and thereby minimizes
multi-user interference, is proposed and analyzed.
This technique is shown to outperform traditional techniques such as
maximum-ratio combining because minimization of interference
power is more critical than maximization of signal power in
the multiple antenna downlink.  Analysis is provided to
quantify the feedback savings, and the technique is seen to
work well with user selection  and is also robust to receiver
estimation error.

\end{abstract}

\section{Introduction}

Multi-user MIMO techniques such as zero-forcing beamforming allow
for simultaneous transmission of multiple data streams even when
each receiver (mobile) has only a single antenna, but very accurate
channel state information (CSI) is generally required at the
transmitter in order to utilize such techniques. In the practically
motivated \textit{finite rate feedback} model, each mobile feeds
back a finite number of bits describing its channel realization at
the beginning of each block or frame. In the \textit{vector}
downlink channel (multiple transmit antennas, single antenna at each
receiver), the feedback bits are determined by quantizing the
channel vector to one of $2^{B}$ quantization vectors. While a
relatively small number of feedback bits suffice to obtain
near-perfect CSIT performance in a point-to-point vector/MISO
(multiple-input, single-output) channel
\cite{Love_Heath_Santipach_Honig}, considerably more feedback is
required in a vector downlink channel. If zero-forcing beamforming
(ZFBF) is used, the feedback rate must be scaled with the number of
transmit antennas as well as SNR in order to achieve rates close to
perfect CSIT systems \cite{Jindal_Finite_BC_Journal}. In such a
system the transmitter emits multiple beams and uses its channel
knowledge to select beamforming vectors such that nulls are created
at certain users.  Inaccurate CSI leads to inaccurate nulling and
thus translates directly into multi-user interference and reduced
SINR/throughput.

In this paper we consider the MIMO downlink channel, in which the transmitter
and each mobile have multiple antennas ($M$ transmit antennas, $N$ antennas per mobile),
in the same limited feedback setting.  We propose a receive antenna combining technique,
dubbed \textit{quantization-based combining} (QBC), that converts the MIMO downlink
into a vector downlink in such a way that the system is able to operate with
reduced channel feedback. Each mobile linearly combines its $N$ antenna outputs
and thereby creates a single antenna channel.  The resulting vector channel is
quantized and fed back, and transmission
is then performed as in a normal vector downlink channel.

With QBC the combiner weights are chosen on the basis of both the channel and
the vector quantization codebook to produce the effective single antenna
channel that can be quantized most accurately.  On the other hand,
traditional combining techniques such as the maximum-ratio based technique
that is optimal for point-to-point MIMO channels with limited channel
feedback \cite{Love_Heath} or direct quantization of the maximum eigenmode
are aimed towards maximization of received signal power but generally
do not minimize channel quantization error.  Since channel quantization
error is so critical in the MIMO downlink channel, quantization-based
combining leads to better performance by minimizing quantization error
(i.e., interference power) possibly at the expense of channel (i.e., signal)
power.

One way to view the advantage of QBC is through its reduced feedback
requirements relative to the vector downlink channel.
In \cite{Jindal_Finite_BC_Journal} it is shown that scaling (per mobile)
feedback as $B = \frac{M-1}{3} P_{dB}$, where $P$ represents the SNR,
suffices to maintain a maximum gap of 3 dB (equivalent to 1 bps/Hz per mobile)
between perfect CSIT and limited feedback performance in a vector downlink
channel employing ZFBF.
With QBC, our analysis shows that the same throughput
(3 dB away from a vector downlink with perfect CSIT) can be achieved if
feedback is scaled at the slower rate of
$B \approx \frac{M-N}{3} P_{dB}$.  In other words, QBC allows
a MIMO downlink to mimic vector downlink performance with reduced
channel feedback.

Alternatively, QBC can be thought of as an effective method to utilize
multiple receive antennas in a downlink channel in the presence of
limited channel feedback.  Although it is possible to send multiple
streams to each mobile if receive combining is not performed, this
requires even more feedback from each mobile than a single-stream approach.
In addition, QBC has the advantage that the transmitter need not be
aware of the number of receive antennas being used.

The remainder of this paper is organized as follows: In
Section \ref{sec-system} we introduce the system model and
some preliminaries.  In Section \ref{sec-selection} we
describe a simple antenna selection method that leads directly
into Section \ref{sec-quantization} where the much more
powerful quantization-based combining technique is described in detail.
In Section \ref{sec-tput} we analyze the throughput and feedback
requirements of QBC.  In Section \ref{sec-compare}
we compare QBC to alternative MIMO downlink techniques,
and finally we conclude in Section \ref{sec-conclusion}.


\section{System Model and Preliminaries} \label{sec-system}

We consider a $K$ mobile (receiver) downlink
channel in which the transmitter (access point) has
$M$ antennas, and each of the mobiles has $N$ antennas.
The received signal at the $i$-th antenna is given by:
\begin{equation}
y_i = {\mathbf h}_i^H {\mathbf x} + n_i, ~~~
   i=1,\ldots,NK
\end{equation}
where ${\bf h}_1, {\bf h}_2, \ldots, {\bf
        h}_{KN}$ are the channel vectors (with ${\bf h}_i \in
        {\mathbb{C}}^{M \times 1}$) describing the $KN$
receive antennas, ${\bf x} \in {\mathbb{C}}^{M \times 1}$ is the
transmitted vector, and ${\bf n}_1, \ldots, {\mathbf n}_{NK}$ are
independent complex Gaussian noise terms with unit variance. The
$k$-th mobile has access to $y_{(k-1)N+1},\ldots, y_{Nk}$. The input
must satisfy a power constraint of $P$, i.e. $E[||{\mathbf x}||^2]
\leq P$. We use $\bH_k$ to denote the concatenation of the $k$-th
mobile's channels, i.e. $\bH_k = [\bh_{(k-1)N+1} \cdots \bh_{Nk}]$.
We consider a block fading channel with iid Rayleigh fading from
block to block, i.e., the channel coefficients are iid complex
Gaussian with unit variance. Each of the mobiles is assumed to have
perfect knowledge of its own channel $\bH_i$, although we analyze
the effect of relaxing this assumption in Section
\ref{sec-rx_error}. In this work we study only the \textit{ergodic
capacity}, or the long-term average throughput. Furthermore, we only
consider systems for which $N < M$ because QBC is not very useful if
$N \geq M$; this point is briefly discussed in Section
\ref{sec-quantization}.

\subsection{Finite Rate Feedback Model} \label{sec-finite}
In the finite rate feedback model, each mobile
quantizes its channel to $B$ bits and feeds back
the bits perfectly and instantaneously to the transmitter at the
beginning of each block \cite{Love_Heath}\cite{Mukkavilli_MIMO}.
Vector quantization is performed using a codebook ${\mathcal C}$
of $2^{B}$ $M$-dimensional unit norm vectors
${\mathcal C} \triangleq \{ \mathbf{w}_1, \ldots, \mathbf{w}_{2^{B}} \}$,
and each mobile quantizes its channel
to the quantization vector that forms the minimum angle to it
\cite{Love_Heath} \cite{Mukkavilli_MIMO}:
\begin{eqnarray}
\hat{\bh}_k
&=& \textrm {arg} \min_{\bw = \bw_1,\ldots,\bw_{2^{B}}}
\sin^2 \left( \angle ({\bf h}_k, \bw) \right). \label{eq-quant}
\end{eqnarray}

For analytical tractability, we study systems using \textit{random
vector quantization} (RVQ) in which each of the $2^{B}$ quantization
vectors is independently chosen from the isotropic distribution on
the $M$-dimensional unit sphere and where each mobile uses an
independently generated codebook \cite{Santipach_Honig}. We analyze
performance averaged over random codebooks; similar to Shannon's
random coding argument, there always exists at least one
quantization codebook that performs as well as the ensemble average.

\subsection{Zero-Forcing Beamforming} \label{sec-zfbf}
After receiving the quantization indices from each of the mobiles,
the AP can use zero-forcing beamforming (ZFBF) to transmit data to
up to $M$ users.  For simplicity let us consider the $N=1$ scenario,
where the channels are the vectors $\bh_1, \ldots, \bh_M$. When ZFBF
is used, the transmitted signal is defined as $\bx = \sum_{k=1}^M
x_k \bv_k $, where each $x_k$ is a scalar (chosen complex Gaussian)
intended for the $k$-th mobile, and $\bv_k \in {\mathcal C}^M$ is
the $k$-th mobile's BF vector. If there are $M$ mobiles (randomly
selected), the beamforming vectors $\bv_1, \ldots, \bv_M$ are chosen
as the normalized rows of the matrix $[ \hat{\bh}_1 \cdots
\hat{\bh}_M]^{-1}$, i.e., they satisfy $||\bv_k|| = 1$ for all $k$
and  $\hat{\bh}_k^H \bv_j = 0$ for all $j \ne k$. If all multi-user
interference is treated as additional noise and equal power loading
is used, the resulting SINR at the $k$-th receiver is given by:
\begin{eqnarray} \label{eq-downlink_sinr}
SINR_k = \frac { \frac{P}{M} | \bh_k^H \bv_k |^2 }
{1 + \sum_{j \ne k}  \frac{P}{M} | \bh_k^H \bv_j |^2 }.
\end{eqnarray}
The coefficient that determines the amount of interference received
at mobile $k$ from the beam intended for mobile $j$, $| \bh_k^H \bv_j |^2$,
is easily seen to be an increasing function of mobile $k$'s quantization error.

In the above expression we have assumed that $M$ mobiles are
randomly selected for transmission and that equal power is allocated to each
mobile.  However, the throughput of zero-forcing based MIMO downlink channels
can be significantly increased by transmitting to an intelligently
selected subset of mobiles \cite{YooGoldsmith06}.  In order to
maximize throughput, users with nearly orthogonal channels and
with large channel magnitudes are selected, and waterfilling can
be performed across the channels of the selected users.
In \cite{DS05} a low-complexity greedy algorithm that selects
users and performs waterfilling is proposed. If this algorithm is
used, a zero-forcing based system can come quite close to the true
sum capacity of the MIMO downlink, even for a moderate number of users.

\subsection{MIMO Downlink with Single Antenna Mobiles} \label{sec-singleant}

In \cite{Jindal_Finite_BC_Journal}
the vector downlink channel ($N=1$)
is analyzed assuming that equal power ZFBF is performed without
user selection on the basis of finite rate feedback (with RVQ).  The basic
result of \cite{Jindal_Finite_BC_Journal} is that:
\begin{eqnarray} \label{eq-rateloss}
R_{FB}(P) \geq R_{CSIT}(P) - \log_2 \left(1 + P \cdot
E \left[\sin^2 \left( \angle(\hat{\bh_k}, \bh_k ) \right) \right] \right)
\end{eqnarray}
where $R_{FB}(P)$ and $R_{CSIT}(P)$ are the ergodic per-user
throughput with feedback and with perfect CSIT, respectively, and the quantity
$E \left[\sin^2 \left( \angle(\hat{\bh_k}, \bh_k ) \right) \right]$ is
the expected quantization error.  The expected quantization error
can be accurately upper bounded by $2^{-\frac{B}{M-1}}$ and therefore the
throughput loss due to limited feedback is upper bounded by
$\log_2 \left(1 + P \cdot 2^{-\frac{B}{M-1}} \right)$, which is
an increasing function of the SNR $P$.  If the number of
feedback bits (per mobile) is scaled with $P$ according to:
\begin{eqnarray*}
B = (M-1) \log_2 P \approx \frac{M-1}{3} P_{dB},
\end{eqnarray*}
then the difference between $R_{FB}(P)$ and $R_{CSIT}(P)$
is upper bounded by $1$ bps/Hz at all SNR's, or equivalently
the power gap is at most 3 dB.
As the remainder of the paper shows, quantization-based
combining significantly reduces the quantization error
(more precisely, it increases the exponential rate at which quantization
error goes to zero as $B$ is increased) and therefore
decreases the rate at which $B$ must be increased as a function of SNR.

\section{Antenna Selection for Reduced Quantization Error}
\label{sec-selection}

In this section we describe a simple antenna selection method that
reduces channel quantization error.
Description of this technique is primarily included for expository reasons,
because the simple concept of antenna selection naturally
extends to the more complex (and powerful) QBC technique.
In point-to-point MIMO, antenna selection corresponds to choosing
the receive antenna with the largest channel gain,
while in the MIMO downlink the receive antenna that can be vector quantized
with minimal angular error is selected.
Mobile 1, which has channel matrix ${\bf H}_1 = [\bh_1 \cdots \bh_N]$ and
a single quantization
codebook consisting of $2^B$ quantization vectors $\bw_1,\ldots,\bw_{2^{B}}$,
first individually quantizes each of its $N$ vector channels $\bh_1, \ldots, \bh_N$
\begin{eqnarray}
\hat{\bf g}_i &=& \textrm {arg} \min_{\bw = \bw_1,\ldots,\bw_{2^{B}}}
\sin^2 \left( \angle ({\bf h}_i, \bw) \right) ~~~~ i=1,\ldots,N,
\end{eqnarray}
and then selects the antenna with the minimum quantization error:
\begin{eqnarray}
j = \textrm {arg} \min_{i = 1,\ldots,N}
\sin^2 \left( \angle ({\bf h}_i, \hat{{\bf g}_i}) \right),
\end{eqnarray}
and feeds back the quantization index corresponding to
$\hat{{\bf g}_j}$.  The mobile uses only antenna $j$ for
reception, and thus the system is effectively transformed
into a vector downlink channel.

Due to the independence of the channel
and quantization vectors, choosing the best of $N$ channel quantizations is
statistically equivalent to quantizing a single vector channel using a
codebook of size $N \cdot 2^{B}$.  Therefore,
antenna selection effectively increases the quantization
codebook size from $2^B$ to $N \cdot 2^{B}$, and
thus the system achieves the same throughput as a vector
downlink with $B + \log_2 N$ feedback bits.
Although not negligible, this advantage is much smaller than that
provided by quantization-based combining.

\section{Quantization-Based Combining}
\label{sec-quantization}

In this section we describe the quantization-based combining (QBC)
technique that reduces channel quantization error by appropriately
combining receive antenna outputs.  We consider a linear
combiner at each mobile, which effectively converts each multiple antenna
mobile into a single antenna receiver.  The combiner structure for a 3 user
channel with 3 transmit antennas ($M=3$) and 2
antennas per mobile ($N=2$) is shown in  Fig. \ref{fig-effective}.
Each mobile linearly combines its $N$ outputs, using
appropriately chosen combiner weights, to produce a scalar
output (denoted by $\yeff_k$).  The effective channel
describing the channel from the transmit antenna array to the
effective output of the $k$-th mobile ($\yeff_k$) is simply
a linear combination of the $N$ vectors describing the
$N$ receive antennas.  After choosing combining weights
the mobile quantizes the effective channel vector
and feeds back the appropriate quantization
index.  Only the effective channel output is used to receive data,
and thus each mobile effectively has only one antenna.

The key to the technique is to \textit{choose combiner weights
that produce an effective channel that can be quantized very accurately};
such a choice must be made on the basis of both
the channel vectors and the quantization codebook.  This is
quite different from maximum ratio combining,
where the combiner weights and quantization vector are chosen such that
received signal power is maximized but quantization error is
generally not minimized.
Note that antenna selection corresponds to choosing the effective channel from
the $N$ columns of ${\bf H}_k$,  while QBC allows for
any linear combination of these $N$ column vectors.


\begin{figure}
    \centering
    \epsfig{file=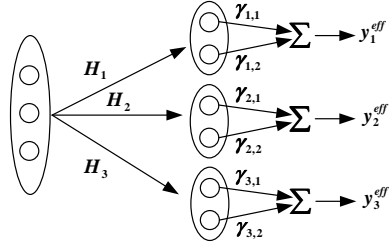, width=2in}
\caption{Effective Channel for $M=K=3$, $N=2$ System}
   \label{fig-effective}
\end{figure}

\subsection{General Description} \label{sec-description}
Let us consider the effective received signal at the first mobile
for some choice of combiner weights, which we denote as
$\boldsymbol{\gamma}_1=(\gamma_{1,1}, \ldots, \gamma_{1,N})$.
In order to maintain a noise variance of one, the combiner weights
are constrained to have unit norm: $||\boldsymbol{\gamma}_1||=1$.  The (scalar)
combiner output, denoted $\yeff_1$, is:
\begin{eqnarray*}
\yeff_1 =
\sum_{i=1}^N  \gamma_{1,i}^H (\bh_i^H \bx + n_i)
&=& \left( \sum_{i=1}^N  \gamma_{1,i}^H \bh_i^H \right) \bx +
\sum_{i=1}^N  \gamma_{1,i}^H n_k \\
&=& (\bheff_1)^H  {\bf x} + n,
\end{eqnarray*}
where $n = \sum_{i=1}^N  \gamma_{1,i}^H n_i$ is
unit variance complex Gaussian because $|\boldsymbol{\gamma}_1|=1$.
The effective channel vector $\bheff_1$
is simply a linear combination of the vectors $\bh_1,\ldots, \bh_N$:
$\bheff_1 = \sum_{i=1}^N  \gamma_{1,i}\bh_i = \bH_1 {\boldsymbol{\gamma}}_1$.
Since $\boldsymbol{\gamma}_1$ can be any unit norm vector, $\bheff_1$
can be in any direction in the $N$-dimensional subspace spanned by $\bh_1,
\ldots, \bh_N$, i.e., in span$(\bH_1)$.\footnote{By
well known properties of iid Rayleigh fading, the matrix ${\bH_1}$
is full rank with probability one \cite{Tulino_Verdu}.}

Because quantization error is so critical to performance, the
objective is to choose combiner weights that yield an effective
channel that can be quantized with minimal error. The error
corresponding to effective channel $\bheff_1$ is
\begin{eqnarray}
\min_{l = 1,\ldots,2^B} \sin^2 \left(
\angle (\bheff_1, \bw_l) \right).
\end{eqnarray}
Therefore, the optimal choice of the effective channel is the solution to:
\begin{eqnarray} \label{eq-double_min}
\min_{\bheff_1} ~~ \min_{l = 1,\ldots,2^B} \sin^2 \left(
\angle (\bheff_1, \bw_l) \right),
\end{eqnarray}
where $\bheff_1$ is allowed to be in any direction in span$(\bH_1)$.
Once the optimal effective channel is determined, the combiner weights
$\boldsymbol{\gamma}_1$ can be determined through a simple pseudo-inverse operation.

Since the expression for the optimum effective channel given in (\ref{eq-double_min})
consists of two minimizations, without loss of optimality
the order of the minimization can be switched to give:
\begin{eqnarray} \label{eq-double_minb}
\min_{l = 1,\ldots,2^B}  ~ \min_{\bheff_1} ~~
\sin^2 \left( \angle (\bheff_1, \bw_l) \right),
\end{eqnarray}
For each quantization vector $\bw_l$, the inner minimization finds
the effective channel vector in span$(\bH_1)$ that forms the minimum
angle with $\bw_l$. By basic geometric principles, the minimizing
$\bheff_1$ is the projection of $\bw_l$ on span$(\bH_1)$. The
solution to the inner minimization in (\ref{eq-double_minb}) is
therefore the sine squared of the angle between $\bw_l$ and its
projection on span$(\bH_1)$, which is referred to as the angle
between $\bw_l$ and the subspace\footnote{If the number of mobile
antennas is equal to the number of transmit antennas ($N = M$), the
channel vectors span $\mathcal{C}^M$ with probability one.
Therefore, each quantization vector has zero angle with the channel
subspace and as a result the solution to the inner minimization in
(\ref{eq-double_minb}) is trivially zero for each $\bw_l$. Thus,
performing quantization with the sole objective of minimizing
angular error (i.e., QBC) is not meaningful when $N=M$ and is
therefore not studied here.}. As a result, the best quantization
vector, i.e., the solution of (\ref{eq-double_minb}), is the vector
that forms the smallest angle between itself and span$(\bH_1)$.  The
optimal effective channel is the (scaled) projection of this
particular quantization vector onto span$(\bH_1)$.

In order to perform quantization, the angle between each
quantization vector and span$(\bH_1)$ must be computed. If ${\bf
q}_1, \ldots, {\bf q}_N$ form an orthonormal basis for span$(\bH_1)$
and ${\bf Q}_1 \triangleq [{\bf q}_1 \cdots {\bf q}_N]$, then
$\sin^2 (\angle( \bw, \textrm{span}(\bH_1)) ) =
1 - ||{\bf Q}_1^H \bw ||^2$.
Therefore, mobile 1's quantized channel, denoted $\hat{\bh_1}$, is:
\begin{eqnarray}
\hat{\bh}_1
= \textrm {arg} \min_{\bw = \bw_1,\ldots,\bw_{2^{B}}}
| \angle( \bw, \textrm{span}(\bH_1)) |
&=& \textrm {arg} \max_{\bw = \bw_1,\ldots,\bw_{2^{B}}} ||{\bf Q}_1^H \bw ||^2.
\end{eqnarray}

Once the quantization vector has been selected, it only remains to
choose the combiner weights.  The projection of $\hat{\bh}_1$ on
span$(\bH_1)$, which is equal to ${\bf Q}_1 {\bf Q}_1^H
\hat{\bh}_1$, is scaled by its norm to produce the unit norm vector
${\bf s}_1^{\textrm{proj}}$. The \textit{direction} specified by
${\bf s}_1^{\textrm{proj}}$ has the minimum quantization error
amongst all directions in span$(\bH_1)$, and therefore the effective
channel should be chosen in this direction. First we find the vector
$\bf{u}_1 \in \mathcal{C}^N$ such that ${\bf H}_1 {\bf u}_1 = {\bf
s}_1^{\textrm{proj}}$,
 and then scale to get  ${\boldsymbol \gamma}_1$.
Since ${\bf s}_1^{\textrm{proj}}$ is in span$(\bH_1)$,
${\bf u}_1$ is uniquely determined by the pseudo-inverse of $\bH_1$:
\begin{eqnarray} \label{eq-u2}
{\bf u}_1 = \left( {\bf H}_1^H {\bf H}_1 \right)^{-1}
{\bf H}_1^H {\bf s}_1^{\textrm{proj}},
\end{eqnarray}
and the combiner weight vector  ${\boldsymbol \gamma}_1$ is the
normalized version of ${\bf u}_1$: $\boldsymbol{\gamma} = \frac {
{\bf u}_1} { ||{\bf u}_1|| }$. The quantization procedure is
illustrated for a $N=2$ channel in Fig. \ref{fig-projection}.  In
the figure the span of the two channel vectors is shown along with
the quantization vector $\bh_1$, its projection on the channel
subspace, and the effective channel.

\subsection{Algorithm Summary}
We now summarize the quantization-based combining procedure performed
at the $k$-th mobile:
\begin{enumerate}

\item Find an orthonormal basis, denoted ${\bf q}_1, \ldots, {\bf q}_N$,
 for span($\bH_k$) and define
${\bf Q}_k \triangleq [{\bf q}_1 \cdots {\bf q}_N]$.

\item Find the quantization vector closest to the channel subspace:
\begin{eqnarray}
\hat{\bh}_k &=& \textrm {arg} \max_{\bw = \bw_1,\ldots,\bw_{2^{B}}} ||{\bf Q}_k^H \bw ||^2.
\end{eqnarray}

\item
Determine the direction of the effective channel by
projecting $\hat{\bh_k}$ onto span($\bH_k$).
\begin{eqnarray}
{\bf s}_k^{\textrm{proj}} = \frac{ {\bf Q}_k {\bf Q}_k^H \hat{\bh_k} }
{ ||{\bf Q}_k {\bf Q}_k^H \hat{\bh_k}|| }.
\end{eqnarray}

\item
Compute the combiner weight vector ${\boldsymbol \gamma}_k$:
\begin{eqnarray} \label{eq-coeff}
{\boldsymbol \gamma}_k = \frac { \left( {\bf H}_k^H {\bf H}_k \right)^{-1} {\bf H}_k^H {\bf s}_1^{\textrm{proj}} }
{ || \left( {\bf H}_k^H {\bf H}_k \right)^{-1} {\bf H}_k^H {\bf s}_1^{\textrm{proj}} || }.
\end{eqnarray}
\end{enumerate}

Each mobile performs these steps, feeds back the index of its
quantized channel $\hat{\bh_k}$, and then linearly combines its $N$
received signals using vector ${\boldsymbol \gamma}_k$ to produce
its effective channel output $y^{\textrm{eff}}_k = (\bheff_k)^H {\bf
x} + n$ with $\bheff_k = \bH_k {\boldsymbol \gamma}_k$. Note that
the transmitter need not be aware of the number of receive antennas
or of the details of this procedure because the downlink channel
appears to be a single receive antenna channel from the
transmitter's perspective; this clearly eases the implementation
burden of QBC.


\begin{figure}
    \centering
    \epsfig{file=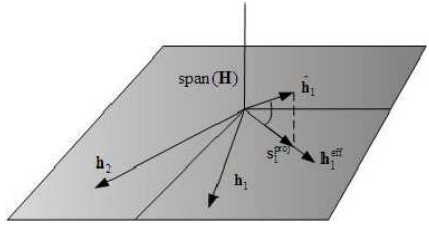, width=2.2in}
\caption{Quantization procedure for a two antenna mobile}
   \label{fig-projection}
\end{figure}

\section{Throughput Analysis} \label{sec-tput}

Quantization-based combining converts the MIMO downlink channel into a
vector downlink with channel vectors
$\bheff_1, \ldots, \bheff_K$ and channel quantizations
$\hat{\bh}_i \cdots \hat{\bh}_K$.
We first derive the statistics of the effective vector channel,
then analyze throughput for ZFBF with equal power
loading and no user selection, and finally quantify the effect of receiver
estimation error.

\subsection{Channel Statistics}
We first determine the distribution of the quantization error and
the effective channel vectors with respect to both the random channels and random
quantization codebooks.

\begin{lemma} \label{lemma-quant_error}
The quantization error $\sin^2 ( \angle( \hat{\bh_k}, \bheff_k ))$,
is the minimum of $2^B$ independent beta $(M-N,N)$ random variables.
\end{lemma}
\begin{proof}
If the columns of $M \times N$ matrix ${\bf Q}_k$ form an
orthonormal basis for span$(\bH_k)$, then $\cos^2 \left(\angle( {\bf
w}_l, \textrm{span}(\bH_k) \right) = ||{\bf Q}_k^H {\bf w}_l||^2$
for any quantization vector. Since the basis vectors and
quantization vectors are isotropically chosen and are independent,
this quantity is the squared norm of the projection of a random unit
norm vector in ${\mathcal C}^M$ onto a random $N$-dimensional
subspace, which is described by the beta distribution with
parameters $N$ and $M-N$ \cite{Beta}. By the properties of the beta
distribution, $\sin^2 \left(\angle( {\bf w}_l, \textrm{span}(\bH_k)
\right) = 1 - \cos^2 \left(\angle( {\bf w}_l, \textrm{span}(\bH_k)
\right)$ is beta $(M-N,N)$. Finally, the independence of the
quantization and channel vectors implies independence of the $2^B$
random variables.
\end{proof}

\begin{lemma} \label{lemma-effective_isotropic}
The normalized effective channels $\frac{\bheff_1}{||\bheff_1||}, \ldots,
\frac{\bheff_K}{||\bheff_K||}$ are iid isotropic vectors in
${\mathcal C}^M$.
\end{lemma}
\begin{proof}
From the earlier description of QBC, note
that $\frac{\bheff_k}{||\bheff_k||} = {\bf s}_k^{\textrm{proj}}$,
which is the projection of the best quantization vector onto
span$(\bH_k)$.  Since each quantization vector is chosen
isotropically, its projection is isotropically distributed within
the subspace.  Furthermore, the best quantization vector is chosen
based solely on the angle between the quantization vector and its
projection.  Thus ${\bf s}_k^{\textrm{proj}}$ is isotropically
distributed in span$(\bH_k)$.  Since this subspace is also
isotropically distributed, the vector ${\bf s}_k^{\textrm{proj}}$ is
isotropically distributed in ${\mathcal C}^M$.  Finally, the independence
of the quantization and channel vectors from mobile to mobile
implies independence of the effective channel directions.
\end{proof}

\begin{lemma} \label{lemma-norm}
The quantity $||\bheff_k||^2$ is $\chi^2_{2(M-N+1)}$.
\end{lemma}
\begin{proof}
Using the notation from Section \ref{sec-description},
the norm of the effective channel is given by:
\begin{eqnarray} \label{eq-effective_norm}
||\bheff_k ||^2 = || \bH_k {\boldsymbol \gamma}_k ||^2 =
|| \bH_k \frac{ {\bf u}_k }{ ||{\bf u}_k|| } ||^2 =
\frac{1}{||{\bf u}_k||^2} || \bH_k {\bf u}_k||^2 =
\frac{|| {\bf s}_k^{\textrm{proj}} ||^2}{|| {\bf u}_k ||^2} = \frac{1}{|| {\bf u}_k ||^2 },
\end{eqnarray}
where we have used the definitions $\bheff_k = \bH_k {\boldsymbol
\gamma}_k$ and ${\boldsymbol \gamma}_k =  \frac{ {\bf u}_k }{ ||{\bf
u}_k|| }$, and the fact that ${\bf u}_k$ satisfies ${\bf H}_k {\bf
u}_k = {\bf s}_k^{\textrm{proj}}$. Therefore, in order to
characterize the norm of the effective channel it is sufficient to
characterize $\frac{1}{|| {\bf u}_k ||^2 }$.  The $N$-dimensional
vector ${\bf u}_k$ is the set of coefficients that allows ${\bf
s}_k^{\textrm{proj}}$, the normalized projection of the chosen
quantization vector, to be expressed as a linear combination of the
columns of $\bH_k$ (i.e., the channel vectors). Because ${\bf
s}_k^{\textrm{proj}}$ is isotropically distributed in span$(\bH_k)$
(Lemma \ref{lemma-effective_isotropic}), if we change coordinates to
any ($N$-dimensional) basis for span$(\bH_k)$ we can assume without
loss of generality that the projection of the quantization vector is
$[1 ~ 0 \cdots 0]^T$.  Therefore, the distribution of $\frac{1}{||
{\bf u}_k ||^2 }$ is the same as the distribution of $\frac{1}{
\left[ \left( {\bf H}_k^H {\bf H}_k \right)^{-1}\right]_{1,1} }$.
Since the $N \times N$ matrix ${\bf H}_k^H {\bf H}_k$ is Wishart
distributed with $M$ degrees of freedom, this quantity is well-known
to be $\chi^2_{2(M-N+1)}$; see \cite{Winters1994} for a proof.
\end{proof}

The norm of the effective channel has
the same distribution as that of a ($M-N+1$)-dimensional random vector
instead of a $M$-dimensional vector. An
arbitrary linear combination (with unit norm) of the $N$ channel
vectors would result in another iid complex Gaussian $M$-dimensional
vector, whose squared norm is $\chi^2_{2M}$, but the weights defining the
effective channel are not arbitrary due to the inverse operation.

\subsection{Sum Rate Performance Relative to Perfect CSIT}

After receiving the quantization indices from each of the mobiles,
a simple transmission option is to perform equal-power ZFBF
based on the channel quantizations (as described in Section \ref{sec-zfbf}).
If $K=M$ or $K>M$ and $M$ users are randomly selected, the resulting SINR at the
$k$-th mobile is given by:
\begin{eqnarray}  \label{eq-sinr_eff}
SINR_k = \frac { \frac{P}{M} | (\bheff_k)^H \bv_k |^2 }
{1 + \sum_{j \ne k}  \frac{P}{M} | (\bheff_k)^H \bv_j |^2 }.
\end{eqnarray}
The ergodic sum rate achieved by QBC, denoted $R_{QBC}(P)$, is
therefore given by:
\begin{eqnarray*}
R_{QBC}(P)
&=&  E_{{\bf H}, \bf{W}} \left[ \log_2 \left(1 +  \frac { \frac{P}{M}
| (\bheff_k)^H \bv_k |^2 } {1 + \sum_{j \ne k}  \frac{P}{M} |
(\bheff_k)^H \bv_j |^2 } \right) \right],
\end{eqnarray*}
where the expectation is taken with respect to the fading
and the random quantization codebooks.

In order to study the benefit of QBC we compare $R_{QBC}(P)$
to the sum rate achieved using zero-forcing beamforming on the basis of
perfect CSIT in an $M$ transmit antenna \textit{vector} downlink channel (single
receive antenna), denoted $R_{ZF-CSIT}(P)$. We use the vector
downlink with perfect CSIT as the benchmark because
QBC converts the system into a vector downlink, and
the rates achieved by QBC cannot exceed $R_{ZF-CSIT}(P)$
(even as $B \rightarrow \infty$).  We later describe how this
metric can easily be translated into a comparison between $R_{QBC}(P)$
and the sum rate achievable with linear precoding (i.e., block diagonalization)
in an $N$ receive antenna MIMO downlink channel with CSIT.

In a vector downlink with perfect CSIT, the BF vectors (denoted ${\bf v}_{ZF,k}$) can be
chosen perfectly orthogonal to all other channels. Thus, the SNR of each user
is as given in (\ref{eq-downlink_sinr}) with zero interference terms
in the denominator and the resulting average rate is:
\begin{eqnarray*}
R_{ZF-CSIT}(P) =  E_{{\bf H}} \left[ \log_2 \left(1 + \frac{P}{M} |{\bf h}_k^H
{\bf v}_{ZF,k}|^2 \right) \right].
\end{eqnarray*}
Following the procedure in \cite{Jindal_Finite_BC_Journal},
the rate gap $\Delta R(P)$ is defined as the difference
between the per-user throughput achieved with perfect CSIT and with feedback-based QBC:
\begin{eqnarray} \label{eq-rate_offset_defn}
\Delta R(P) \triangleq R_{ZF-CSIT}(P) - R_{QBC}(P).
\end{eqnarray}
Similar to Theorem 1 of \cite{Jindal_Finite_BC_Journal}, we can upper bound
this throughput loss:
\begin{theorem} \label{thm-gap}
The per-user throughput loss is upper bounded by:
\begin{eqnarray*}
\Delta R(P) &\leq& \left( \sum_{l=M-N+1}^{M-1} \frac{1}{l}
\right)  \log_2 e + \log_2 \left( 1\! +\! P\! \left( \frac{M\!-N\!+1}{M}
\right) \! E[\sin^2 ( \angle( \hat{\bh_k}, \bheff_k ))]
 \right)
\end{eqnarray*}
\end{theorem}
\begin{proof}
See Appendix.
\end{proof}
The first term in the expression is the throughput loss due to the
reduced norm (Lemma  \ref{lemma-norm}) of the effective channel,
while the second (more significant) term,
which is an increasing function of $P$, is due to quantization error.
In order to quantify this rate gap, the expected quantization error
needs to be bounded.    By Lemma
\ref{lemma-quant_error}, the quantization error is the minimum of
$2^B$ iid beta$(M-N, N)$ RV's.  Furthermore, a general
result on ordered statistics applied to beta RV's gives
\cite[Chapter 4.I.B]{Beta}:
\begin{eqnarray*}
E[\sin^2 ( \angle( \hat{\bh_k}, \bheff_k ))] \leq F^{-1}_X \left(
2^{-B} \right)
\end{eqnarray*}
where $F_X(x)$ is the inverse of the CDF of a beta $(M-N, N)$ random
variable, which is:
\begin{eqnarray*}
F_X(x) &=& \sum_{i=0}^{N-1} {M \! - \! 1 \choose N \! - \! 1 \! - \!
i} x^{M-N+i} (1-x)^{N-1+i} \approx {M \! - \! 1 \choose N \! - \! 1}
x^{M-N},
\end{eqnarray*}
where the approximation is the result of keeping only the lowest order $x$ term and dropping
$(1-x)$ terms; this is valid for small values of $x$.
Using this we get the
following approximation:
\begin{eqnarray} \label{eq-approx_quant}
E[\sin^2 ( \angle( \hat{\bh_k}, \bheff_k ))] \approx
2^{-\frac{B}{M-N}} {M \! - \! 1 \choose N \! - \!
1}^{-\frac{1}{M-N}}.
\end{eqnarray}
The accuracy of this approximation is later verified by our
numerical results. Plugging this approximation into the upper bound
in Theorem \ref{thm-gap} we get:
\begin{eqnarray} \label{eq-gap_approx}
\Delta R (P) &\approx&   \left(\sum_{l=M-N+1}^{M-1} \frac{1}{l}
\right)\log_2 e + \log_2 \left( 1 \! + \! P \cdot \left( \frac{M \! - \! N \!
+ \! 1}{M} \right)  2^{-\frac{B}{M-N}} {M\!-\!1 \choose
N\!-\!1}^{-\frac{1}{M-N}} \right)
\end{eqnarray}

If $B$ is fixed, quantization error causes the system to become interference-limited
as the SNR is increased (see \cite[Theorem
2]{Jindal_Finite_BC_Journal} for a formal proof when $N=1$).
However, if $B$ is scaled with the SNR $P$
such that the quantization error decreases as $\frac{1}{P}$, the rate
gap in (\ref{eq-gap_approx}) can be kept constant and
the full multiplexing gain ($M$) is achieved. In order to determine this scaling, we set the
approximation of $\Delta R(P)$ in (\ref{eq-gap_approx}) equal to
a rate constant $\log_2 b$ and solve for $B$ as a function of $P$.  Thus,
a per-mobile rate loss of at most $\log_2 b$ (relative to $R_{ZF-CSIT}(P)$)
is maintained if $B$ is scaled as:
\begin{eqnarray}
B_N  &\approx& (M-N) \log_2 P - (M-N) \log_2 c
 - (M-N) \log_2 \left(\frac{M}{M \! - \! N \! + \! 1} \right) - \log_2  {M \! - \! 1 \choose N \! - \! 1} ,
\nonumber \\
 &\approx& \frac{M-N}{3} P_{dB} - (M-N) \log_2 c
 - (M-N) \log_2 \left(\frac{M}{M \! - \! N \! + \! 1} \right) - \log_2  {M \! - \! 1 \choose N \! - \! 1},
\label{eq-scaling_N}
\end{eqnarray}
where $c = b \cdot e^{-(\sum_{l=M-N+1}^{M-1} \frac{1}{l})} - 1$.
Note that a per user rate gap of $\log_2 b =1$ bps/Hz is
equivalent to a 3 dB power gap in the sum rate curves.

As discussed in
Section \ref{sec-singleant}, scaling feedback in a single receive
antenna downlink as $B_1 = \frac{M-1}{3} P_{dB}$
maintains a 3 dB gap from perfect CSIT throughput.
Feedback must also be increased linearly if QBC is used,
but the slope of this increase is $\frac{M-1}{3}$ when mobiles
have only a single antenna compared to a slope of
$\frac{M-N}{3}$ for antenna combining.
If we compute the difference between the $N=1$ feedback
load and the QBC feedback load, we can quantify
\textit{how much less feedback is required to achieve the same throughput
(3 dB away from a vector downlink channel with perfect CSIT)
if QBC is used with $N$ antennas/mobile}:
\begin{eqnarray*}
\Delta_{QBC}(N) = B_1 - B_N
\approx \frac{N-1}{3} P_{dB}  +\log_2  {M \! - \! 1 \choose N \! - \! 1}
- (N-1) \log_2 e.
\end{eqnarray*}

The sum rate of a 6 transmit antenna downlink channel ($M=6$)
is plotted in Fig. \ref{fig-sum_rate}.  The perfect CSIT
zero-forcing curve is plotted along with the rates
achieved using finite rate feedback with $B$
scaled according to (\ref{eq-scaling_N}) for $N=1,2$ and $3$.
For $N=2$ and $N=3$ QBC is performed and the fact that the throughput
curves are approximately 3 dB away from the perfect
CSIT curve verify the accuracy of the approximations
used to derive the feedback scaling expression in
(\ref{eq-scaling_N}).  In this system,
the feedback savings at 20 dB are 7 and 12 bits, respectively,
for $2$ and $3$ receive antennas.  All numerical results
in the paper are generated using the method described in
Appendix \ref{sec-numerical_method}.

It is also important to compare QBC throughput to
the throughput of a MIMO downlink channel with $N$ antennas per mobile.
The most meaningful comparison is
to the rate achievable with block diagonalization (BD) \cite{Spencer2004}
without user selection and with equal power loading.  In this case, $\frac{M}{N}$ mobiles are
transmitted to (with $N$ data streams per mobile).  In
\cite{LeeJindal07IT} it is shown that the BD sum rate is
\begin{eqnarray*}
\Delta_{BD-ZF}(N) = (\log_2 e) \frac{M}{N} \sum_{j=1}^{N-1} \frac{N-j}{j}
\end{eqnarray*}
larger than $R_{ZF-CSIT}(P)$ at asymptotically high SNR, and that
this offset is accurate even for moderate SNR's.  This can
be translated to a power offset by multiplying by $\frac{3}{M}$
to give $\frac{3 \log_2 e}{N} \sum_{j=1}^{N-1} \frac{N-j}{j}$ dB,
which equates to 2.16 dB and 3.61 dB for $N=2$ and $N=3$.
Therefore, the rate offset between QBC and BD with CSIT is the
sum of $\Delta R(P)$ (equation \ref{eq-rate_offset_defn}) and
$\Delta_{BD-ZF}(N)$.
In Fig. \ref{fig-sum_rate} the BD sum rate curves
are plotted, and their shifts relative to ZF-CSIT are seen to
follow the predicted power gaps.


\begin{figure}
    \centering
    \epsfig{file=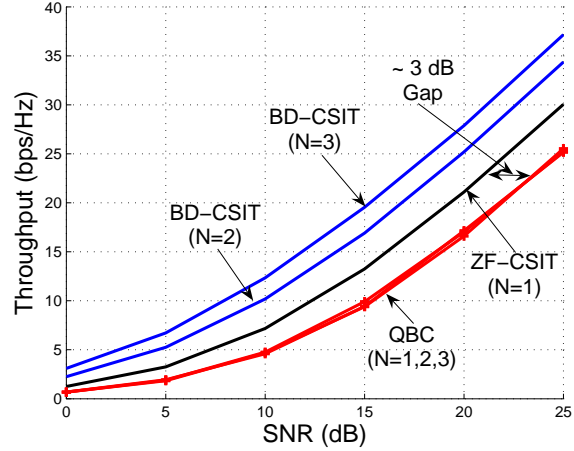, width=3.35in}
\caption{Sum rate of $M=K=6$ downlink channel}
   \label{fig-sum_rate}
\end{figure}

\subsection{Effect of Receiver Estimation Error} \label{sec-rx_error}

Although the analysis until now has assumed perfect CSI at the
mobiles, a practical system always has some level of receiver error.
We consider the scenario where a shared pilot sequence is used to
train the mobiles.  If $\beta M$ downlink pilots are used ($\beta
\geq 1$ pilots per transmit antenna), channel estimation at the
$k$-th mobile is performed on the basis of observation $\bG_k =
\sqrt{\beta P} \bH_k + {\bf n}_k$. The MMSE estimate of $\bH_k$ is
$\hat{\bG_k} = \frac{ \sqrt{\beta P} }{1 + \beta P} \bG_k$, and the
true channel matrix can be written as the sum of the MMSE estimate
and independent estimation error:
\begin{eqnarray} \label{eq-HRX}
\bH_k = \hat{\bG_k} + {\bf e}_k,
\end{eqnarray}
where ${\bf e}_k$ is white Gaussian noise, independent of the
estimate $\hat{\bG_k}$, with per-component variance $(1 + \beta
P)^{-1}$. After computing the channel estimate $\hat{\bG_k}$, the
mobile performs QBC on the basis of the estimate $\hat{\bG_k}$ to
determine the combining vector ${\boldsymbol \gamma}_k$. As a
result, the quantization vector $\hat{\bh_k}$ very accurately
quantizes the vector $\hat{\bG_k}{\boldsymbol{\gamma}}_k$, which is
the mobile's estimate of the effective channel output, while the
actual effective channel is given by $\bheff_k = \bH_k
{\boldsymbol{\gamma}}_k$.

For simplicity we assume that coherent communication is possible,
and therefore the long-term average throughput is again $E[\log_2(1
+ SINR_k)]$ where the same expression for SINR given in
(\ref{eq-sinr_eff}) applies\footnote{We have effectively assumed
that each mobile can estimate the phase and SINR at the effective
channel output. In practice this could be accomplished via a second
round of pilots as described in \cite{Caire_Jindal_ISIT07}.}. The
general throughput analysis in Section \ref{sec-tput} still applies,
and in particular, the rate gap upper bound given in Theorem
\ref{thm-gap} still holds if the expected quantization error takes
into account the effect of receiver noise.  As shown in Appendix
\ref{sec-rategap_rx_error}, the approximate rate loss with receiver
error is:
\begin{eqnarray} \label{eq-gap_approx2}
\Delta R (P) &\approx&  \log_2 e \left(\sum_{l=M-N+1}^{M-1} \frac{1}{l}
\right) + \log_2 \left( 1 \! + \! P \cdot \left( \frac{M \! - \! N \!
+ \! 1}{M} \right)  2^{-\frac{B}{M-N}} {M\!-\!1 \choose
N\!-\!1}^{-\frac{1}{M-N}} + \frac{1}{\beta} \right).
\end{eqnarray}
Comparing this expression to (\ref{eq-gap_approx}) we see that
estimation error leads only to the introduction of an additional
$\frac{1}{\beta}$ term. If feedback is scaled according to
(\ref{eq-scaling_N}) the rate loss is $\log_2(b + \beta^{-1})$
rather than $\log_2 (b)$. In Figure \ref{fig-comb_RXerror} the
throughput of a 4 mobile system with $M=4$ and $N=2$ is plotted for
perfect CSIT/CSIR and for QBC performed on the basis of perfect CSIR
($\beta = \infty$) and imperfect CSIR for $\beta=1$ and $\beta=2$.
Estimation error causes non-negligible degradation, but the loss
decreases rather quickly with $\beta$ (which can be increased at a
reasonable resource cost because pilots are shared).


\begin{figure}
    \centering
    \epsfig{file=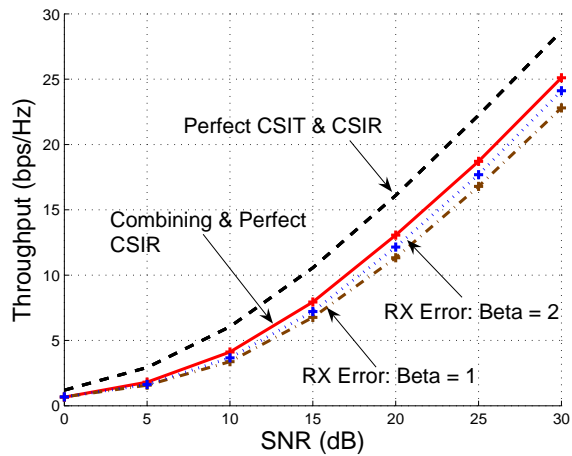, width=3.35in}
\caption{Combining with Imperfect CSIR: $M=4$, $N=2$, $K=4$, $B$
scaled with SNR}
   \label{fig-comb_RXerror}
\end{figure}

\section{Performance Comparisons} \label{sec-compare}

In this section we compare the throughput of QBC to other receive
combining techniques and to limited feedback-based block
diagonalization\footnote{ It should be noted that comparisons with
block diagonalization are somewhat rough because systems that
perform BD on the basis of limited feedback and that employ
user/stream selection have not yet been extensively studied in the
literature, to the best of our knowledge. As a result, it may be
possible to improve upon the BD systems we use here as the point of
comparison.}.  For all results on receiving combining, the user
selection algorithm of \cite{DS05} is applied assuming limited
feedback ($B$ bits) regarding the direction of the effective channel
and perfect knowledge of the effective channel
norm\footnote{Although the rate gap upper bound derived in Theorem
\ref{thm-gap} only rigorously applies to systems with equal power
loading and random selection of $M$ mobiles, the bound can be used
to reasonably approximate the throughput degradation due to limited
feedback even when user selection is performed.  See
\cite{Yoo_Jindal_Goldsmith07} for a further discussion of the effect
of limited feedback on systems employing user selection.}. We first
describe these alternative approaches and then discuss some
numerical results.

\subsection{Alternate Combining Techniques}

The optimal receive combining technique for a point-to-point MIMO channel
in a limited feedback setting is to select the quantization vector that maximizes
received power \cite{Love_Heath}:
\begin{eqnarray} \label{eq-quant_mrc}
\hat{\bh}_k &=& \textrm {arg} \max_{\bw=\bw_1,\ldots,\bw_{2^{B}}}
|| \bH_k^H \bw ||^2. \end{eqnarray}
Because this method roughly corresponds to maximum ratio combining, it is referred to as MRC.
If BF vector $\bw$ is used by the transmitter, received power is maximized by choosing
$\boldsymbol{\gamma} = \frac{ \bH_k^H \bw }{|| \bH_k^H \bw ||}$
\cite{Love_Heath}, which yields
$\bheff_k = \bH_k \boldsymbol{\gamma}_k = \frac{ \bH_k \bH_k^H \bw_k}{ || \bH_k^H \bw_k||}$.
When $B$ is not very small, with high
probability the quantization vector that maximizes $|| \bH_k^H \bw ||^2$
is the vector that is closest to the eigenvector corresponding
to the maximum eigenvalue of $\bH_k \bH_k^H$.
To see this, consider the maximization of $|| \bH_k^H \bw ||$ when
$\bw$ is constrained to have unit norm but need not be selected
from a finite codebook.  This corresponds to the classical definition
of the matrix norm, and the optimizing $\bw$ is
in the direction of the maximum singular value of $\bH_k$. When
$B$ is not too small, the quantization error is very small and as a
result the solution to (\ref{eq-quant_mrc}) is extremely close
to $||\bH_k||^2$.
As a result, \textit{selecting the quantization vector according to the
criteria in (\ref{eq-quant_mrc}) is roughly equivalent to
directly finding the quantization vector that is closest to the
direction of the maximum singular value of $\bH_k$.}

The maximum singular value of $\bH_k$ can be directly quantized if
the mobile first selects the combiner weights
$\boldsymbol{\gamma}_k$ such that the effective channel
$\bheff_k = \bH_k \boldsymbol{\gamma}_k$ is in the direction
of the maximum singular value, which corresponds to selecting
$\boldsymbol{\gamma}_k$ equal to the eigenvector corresponding
to the maximum eigenvalue of the $N \times N$ matrix $\bH_k^H \bH_k$,
and then finds the quantization vector closest to $\bheff_k$.
The effective channel norm satisfies
$||\bheff_k||^2 = ||\bH_k||^2$, which can be reasonably
approximated as a scaled version of a $\chi^2_{2MN}$ random
variable \cite{Paulraj_Gore_Nabar}.  Therefore the norm of the effective
channel is large, but notice that the quantization procedure reduces
to standard vector quantization, for which the error is roughly
$2^{-\frac{B}{M-1}}$.

In Figure \ref{fig-comb_compare}, numerically computed values
of the quantization error
($\log_2 (E[\sin^2 ( \angle(\bheff_k, \hat{\bh_k}))]$)
are shown for QBC, antenna
selection, MRC (corresponding to equation
\ref{eq-quant_mrc}), and direct quantization of the maximum
eigenvector, along with approximation $2^{-\frac{B}{M-1}}$ as well
as the approximation from (\ref{eq-approx_quant}), for a $M=4$, $N=2$ channel.
Note that the error of QBC is very well
approximated by (\ref{eq-approx_quant}), and the exponential rate
of decrease of the other techniques are all well approximated
by $2^{-\frac{B}{M-1}}$.

\begin{figure}
    \centering
    \epsfig{file=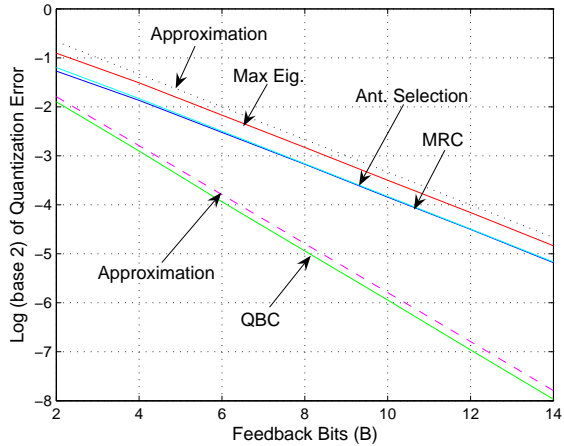, width=3.35in}
\caption{Quantization Error for Different Combining Techniques
($M=4$, $N=2$)}
   \label{fig-comb_compare}
\end{figure}

Each combining technique transforms the MIMO downlink into a vector
downlink with a modified channel norm and quantization error. These
techniques are summarized in Table \ref{table-combining}. The key
point is that only QBC changes the exponent of the quantization
error\footnote{An improvement over QBC is to choose the quantization
vector and combining weights that maximize the expected received
SINR (the true SINR depends on the BF vectors, which are unknown to
the mobile). This extension of QBC, which will surely outperform QBC
and MRC, has been under investigation by other researchers since the
initial submission of this manuscript and the results will be
published shortly \cite{TrivHuangBoc07_Asilomar}.}, which determines
the rate at which feedback increases with SNR. When comparing these
techniques note that the complexity of QBC and MRC are essentially
the same: QBC and MRC require computation of $|| {\bf Q}_k^H \bw
||^2$ and $|| \bH_k^H \bw ||^2$, respectively.

\begin{table}
\begin{center}
\begin{tabular}{l|l|l}
& Effective Channel Norm & Quantization Error \\ \hline Single RX
Antenna ($N=1$) & $\chi^2_{2M}$ & $2^{-B/(M-1)}$ \\ \hline Antenna
Selection  & $\chi^2_{2M}$ & $2^{-(B+\log_2 N)/(M-1)}$ \\ \hline MRC
& $\approx$ max eigenvalue & $2^{-B/(M-1)}$ \\ \hline Max
Eigenvector  & max eigenvalue & $2^{-B/(M-1)}$ \\ \hline QBC  &
$\chi^2_{2(M-N+1)}$ & $2^{-B/(M-N)}$ \\ \hline
\end{tabular}
\caption{Summary of Combining Techniques} \label{table-combining}
\end{center}
\end{table}

\subsection{Block Diagonalization} \label{sec-bd}

An alternative manner in which multiple receive antennas can be used is to
extend the linear precoding structure of ZFBF to allow for
transmission of multiple data streams to each mobile.
Block diagonalization (BD) selects precoding matrices such multi-user interference
is eliminated at each receiver, similar to ZFBF. In order to select appropriate
precoding matrices, the transmitter must know
the $N$-dimensional subspace spanned by each mobile channel ${\bf H}_k$.
Thus an appropriate feedback strategy is to have each
mobile quantize and feedback its channel subspace.  The effect of
limited feedback in this setting (assuming there are $\frac{M}{N}$ mobiles
and equal power loading across users and streams is performed)
was studied in \cite{RavindranJindal07_ICASSP}.  In order to
achieve a bounded rate loss relative to a perfect CSIT (BD) system, feedback
(per mobile) needs to scale approximately as $N(M-N) \log_2 P$.  Thus,
the aggregate feedback load summed over $\frac{M}{N}$ mobiles is
approximately $M (M-N) \log_2 P$, which is (approximately) the same as
the aggregate feedback in a QBC system in which each of the $M$ mobiles
uses $B \approx (M-N) \log_2 P$.  Thus, there is a rough equivalence between
QBC and BD in terms of feedback scaling, and this is later confirmed by
our numerical results.

It is also possible to perform user and stream selection when BD is
used, and \cite{BoccardiHuang07_ICASSP} presents an extension of the
algorithm of \cite{DS05} to the multiple receive antenna setting
(referred to as maximum eigenmode transmission, or MET).  In
essence, MET treats each mobile's $N$ eigenmodes as a different
single antenna receiver and selects eigenmodes in a greedy fashion
using the approach of \cite{DS05}.  Thus, in a limited feedback
setting a reasonable strategy is to have each user separately
quantize the directions of its $N$ eigenvectors and also feed back
the corresponding eigenvalues.

\subsection{Numerical Results} \label{sec-numerical}

In Figures \ref{fig-comb_scaledB} and  \ref{fig-comb_fixedB}
throughput curves are shown for a 4 transmit antenna, 2 receive antenna ($M=4$, $N=2$)
system with $K=4$ mobiles.  Sum rate is plotted for three different combining techniques
(QBC, antenna selection, and MRC) and for a vector downlink channel ($N=1$);
the BD curves are discussed in later paragraphs.  In  Fig.
\ref{fig-comb_scaledB}, $B$ (per mobile) is scaled according to (\ref{eq-scaling_N}),
i.e., roughly as $(M-N) \log_2 P$, while in Fig.  \ref{fig-comb_fixedB}
each mobile uses 10 bits of feedback.
As expected, the throughput of antenna selection, MRC, and the single antenna
system all lag behind QBC in Fig. \ref{fig-comb_scaledB}, particularly
at high SNR.  This is because the $(M-N) \log_2 P$ scaling of
feedback is simply not sufficient to maintain good performance
if these techniques are used.  To be more precise, the quantization
error goes to zero slower than $\frac{1}{P}$ which corresponds to interference
power that increases with SNR, and thus a reduction in the slope
(i.e., multiplexing gain) of these curves.
In Fig. \ref{fig-comb_fixedB}, MRC outperforms QBC for SNR less than approximately
12 dB because signal power is more important than quantization error (i.e., interference
power), i.e., the system is not yet interference-limited.  However, at higher SNR's
QBC outperforms MRC because of the increased importance of
quantization error.


\begin{figure}
    \centering
    \epsfig{file=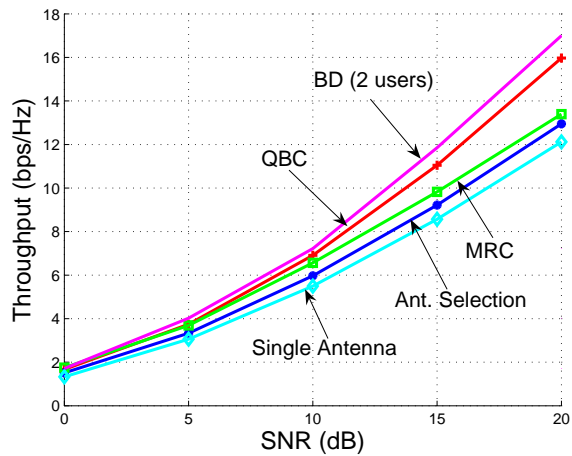, width=3.35in}
\caption{Different Combining Techniques: $M=4$, $N=2$, $K=4$, $B$
scaled with SNR}
   \label{fig-comb_scaledB}
\end{figure}

\begin{figure}
    \centering
    \epsfig{file=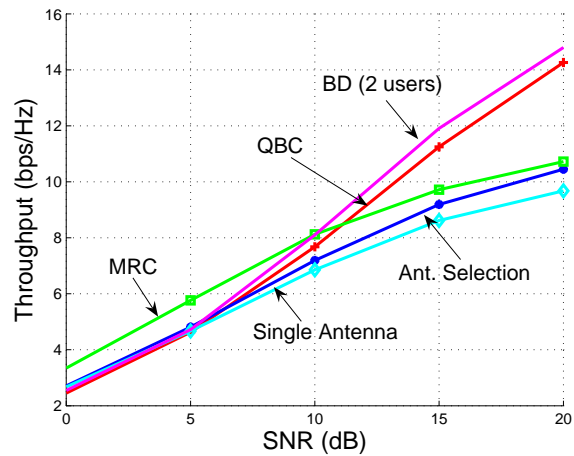, width=3.35in}
\caption{Different Combining Techniques: $M=4$, $N=2$, $K=4$,
$B=10$}
   \label{fig-comb_fixedB}
\end{figure}

Figures \ref{fig-comb_scaledB} and  \ref{fig-comb_fixedB} also include
plots of the throughput of a BD system. In this system, 2 of the 4 users
are randomly selected to feedback subspace information, and equal power BD with no selection
is used to send 2 streams to each of these mobiles, for a total of 4 streams.
In order to equalize the aggregate feedback load, each of the 2 users is allocated double the
feedback budget of the combining-based systems; this corresponds to using
two times the scaling of (\ref{eq-scaling_N}) in Fig. \ref{fig-comb_scaledB}
and 20 bits per mobile in Fig.  \ref{fig-comb_fixedB}.  BD performs
slightly better than QBC in both figures, but we later see that this
advantage is lost for larger $K$.

Figures \ref{fig-select_M4_B10} displays throughput for a 4 transmit antenna,
2 receive antenna ($M=4$, $N=2$) system at 10 dB against $K$, the number of mobiles.
Capacity refers to the sum capacity of the system (with CSIT), MET-CSIT is the
throughput achieved using the MET algorithm  on the basis of
CSIT\cite{BoccardiHuang07_ICASSP}, and ZF-CSIT is the throughput of a vector downlink
with CSIT and user selection \cite{DS05}.  Below these are four limited feedback
curves for 10 bits of feedback per mobile.  The first three, QBC, MRC, and antenna
selection, correspond to different combining techniques, while MET-FB corresponds
to performing MET on the basis of 5 bit quantization of each eigenmode
(10 bits total feedback per mobile).  QBC achieves significantly
higher throughput than MRC or antenna selection, particularly for larger
values of $K$.  The ZF-CSIT curve is shown because it serves as an upper bound on
the performance of QBC, and the gap between the two
is quite reasonable even for $B=10$.
MET-FB is seen to perform extremely poorly:  this is not too surprising because
the MET algorithm is likely to only choose the strongest eigenmode of a few users
\cite{BoccardiHuang07_ICASSP}, and thus half of the feedback is essentially
wasted on quantization of each user's weakest eigenmode.  This motivates
dedicating all 10 bits to quantization of the strongest eigenmode, but note that this
essentially corresponds to MRC, which is outperformed by QBC.  The huge gap between MET-CSIT and
MET-FB indicates that MET has the potential to provide excellent performance,
but extremely high levels of feedback may be necessary to realize MET's potential.

Finally, Figure \ref{fig-select_M6_N12} shows throughput versus number of users
$K$ for a 6 transmit antenna ($M=6$) channel with either 1 or 2 receive
antennas.  Sum capacity for $N=1$ and $N=2$ is plotted, along with the
sum rate of a perfect-CSIT TDMA system in which only the receiver with the largest
point-to-point capacity is selected for transmission.  The ZF and QBC curves correspond
to systems with user selection and either single receive antennas or quantization-based
combining, respectively, for feedback levels of 10, 15, and 20 bits per mobile.
For each feedback level, an additional receive antenna with QBC provides a significant
throughput gain relative to a single receive antenna system.  Furthermore, QBC
significantly outperforms TDMA ($N=2$) for $B=15$ or $B=20$, and provides an advantage
over TDMA for $B=10$ when the number of users is sufficiently large.
Note, however, that there is a significant gap between QBC and $N=2$ capacity
even when 20 bits of feedback are used; this indicates that there may be room for
significant improvement beyond QBC.


\begin{figure}
    \centering
    \epsfig{file=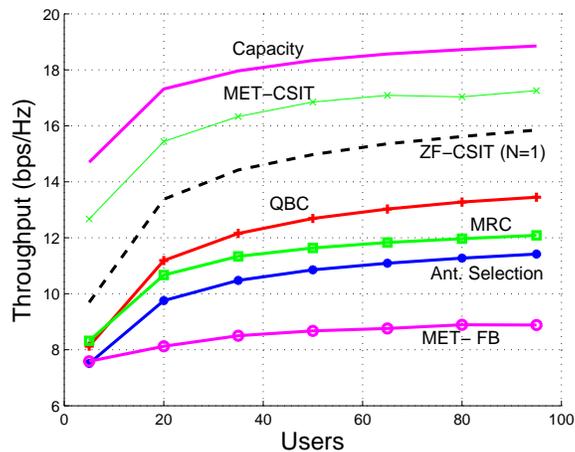, width=3.35in}
\caption{Combining and User Selection: $M=4$, $N=2$, $B=10$}
   \label{fig-select_M4_B10}
\end{figure}

\begin{figure}
    \centering
    \epsfig{file=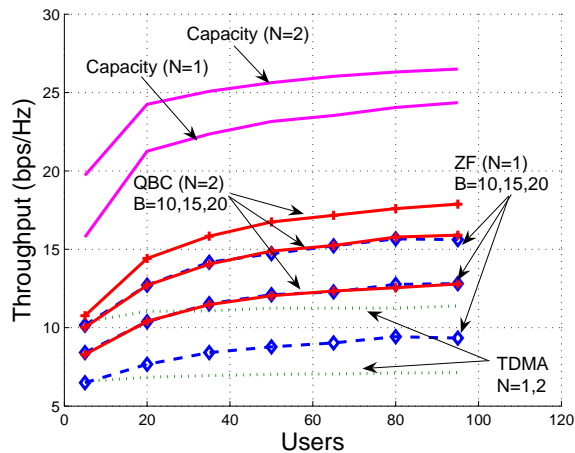, width=3.35in}
\caption{Combining and User Selection: $M=6$, $N=1,2$}
   \label{fig-select_M6_N12}
\end{figure}

\section{Conclusion} \label{sec-conclusion}

The performance of multi-user MIMO techniques such as
zero-forcing beamforming critically depend on the accuracy of the channel
state information provided to the transmitter.  In this paper,
we have shown that receive antenna combining can be used to
reduce channel quantization error in limited feedback MIMO downlink channels,
and thus significantly reduce channel feedback requirements.
Unlike traditional maximum-ratio combining techniques that
maximize received signal power, the proposed quantization-based
combining technique minimizes quantization error, which translates
into minimization of multi-user interference power.

Antenna combining is just one method by which
multiple receive antennas can be used in the MIMO downlink.  It
is also possible to transmit multiple streams to each mobile, or to
use receive antennas for interference cancellation if the structure
of the transmitted signal is known to the mobile.  It remains to
be seen which of these techniques is most beneficial in practical
wireless systems when channel feedback resources and complexity
requirements are carefully accounted for.

\appendices

\section{Proof of Theorem \ref{thm-gap}}
Plugging the rate expressions into the definition of $\Delta(P)$, we
have $\Delta(P) = \Delta_a + \Delta_b$ where
\begin{eqnarray*}
\Delta_a &=& E_{{\bf H}} \left[ \log_2 \left(1 + \rho |{\bf
h}_k^H {\bf v}_{ZF,k}|^2 \right) \right] - E_{{\bf H}, \bf{W}}
\left[ \log_2
\left(1 + \sum_{j=1}^M \rho | (\bheff_k)^H \bv_j |^2 \right) \right] \\
\Delta_b &=& E_{{\bf H}, \bf{W}} \left[ \log_2 \left(1 + \sum_{j \ne
k} \rho | (\bheff_k)^H \bv_j |^2 \right) \right],
\end{eqnarray*}
where $\rho \triangleq \frac{P}{M}$.
To upper bound $\Delta_a$, we define normalized vectors
$\tilde{\bh_k} = \bh_k / || \bh_k ||$ and $\tilde{\bheff_k} =
\bheff_k / || \bheff_k ||$, and note that the norm and directions of
$\bh_k$ and of $\bheff_k$ are independent. Using this we have:
\begin{eqnarray}
E_{{\bf H}, \bf{W}} \left[ \log_2 \left(1 + \sum_{j=1}^M \rho |
(\bheff_k)^H \bv_j |^2 \right) \right] &\geq&  E_{{\bf H}, \bf{W}}
\left[ \log_2 \left(1 + \rho | (\bheff_k)^H \bv_k |^2 \right)
\right] \nonumber
\\ &=& E_{{\bf H}, \bf{W}} \left[ \log_2 \left(1 +
\rho || \bheff_k ||^2 |
\tilde{\bheff_k}^H \bv_k |^2 \right) \right] \nonumber \\
&=& E_{{\bf H}} \left[ \log_2 \left(1 + \rho X_{\beta}
||\bh_k||^2 | \tilde{\bh_k}^H {\bf v}_{ZF,k}|^2 \right) \right],
\label{eq-a1}
\end{eqnarray}
where $X_{\beta}$ is $\beta(M-N+1,N-1)$. Since the BF
vector ${\bf v}_{ZF,k}$ is chosen orthogonal to the $(M-1)$ other
channel vectors $\{ \bh_j \}_{j \ne k}$, each of which is an iid
isotropic vector, it is isotropic and is \textit{independent} of
$\tilde{\bh_k}$.  By Lemma
\ref{lemma-effective_isotropic} the same is also true of $\bv_k$ and
$\tilde{\bheff_k}$, and therefore we can substitute $|
\tilde{\bh_k}^H {\bf v}_{ZF,k}|^2$ for $| (\bheff_k)^H \bv_k |^2$.
Finally, note that the product $X_{\beta} ||\bh_k||^2$ is
$\chi^2_{2(M-N+1)}$ because $||\bh_k||^2$ is $\chi^2_{2M}$, and
therefore $X_{\beta} ||\bh_k||^2$ and $|| \bheff_k ||^2$ have the
same distribution.
Using (\ref{eq-a1}) we get:
\begin{eqnarray*}
\Delta_a
\leq E_{{\bf H}} \left[ \log_2 \left( \frac{ 1 + \rho
||\bh_k||^2 | \tilde{\bh_k}^H {\bf v}_{ZF,k}|^2}{1 + \rho X_{\beta}
 ||\bh_k||^2 | \tilde{\bh_k}^H {\bf v}_{ZF,k}|^2 }
\right) \right]
\leq - E \left[ \log_2 \left( X_{\beta} \right) \right]
= \log_2 e \left( \sum_{l=M-N+1}^{M-1} \frac{1}{l} \right),
\end{eqnarray*}
where we have used $ \log_2 \left( X_{\beta} \right) =
\log_2 \left( \frac{\chi^2_{2M}}{ \chi^2_{2(M-N+1)}} \right)$
and results from \cite{Tulino_Verdu} to to compute $E \left[ \log_2 \left( X_{\beta} \right) \right]$.

Finally, we upper bound $\Delta_b$ using Jensen's inequality:
\begin{eqnarray*} \Delta_b &\leq& \log_2 \left(1 + E \left[ \sum_{j \ne k} \rho | (\bheff_k)^H
\bv_j |^2 \right] \right) \\
 &=& \log_2 \left(1 +   \rho (M-1) E \left[||(\bheff_k)||^2 \right]
E \left[ |(\tilde{\bheff}_k)^H \bv_j |^2 \right] \right) \\
&=& \log_2 \left(1 +  \rho (M-1)(M-N+1) E \left[
|(\tilde{\bheff}_k)^H \bv_j |^2
\right] \right) \\
&=& \log_2 \left(1 +  \rho (M-N+1) E \left[ \sin^2 \left(
\angle \left(\tilde{\bheff}_k, \bh_k \right) \right) \right] \right),
\end{eqnarray*}
where the final step uses Lemma 2 of \cite{Jindal_Finite_BC_Journal}
to get $E \left[ |(\tilde{\bheff}_k)^H \bv_j |^2 \right] = \frac{1}{M-1} E
\left[ \sin^2 \left( \angle \left(\tilde{\bheff}_k, \bh_k \right) \right)
\right]$.

\section{Generation of Numerical Results} \label{sec-numerical_method}
Rather than performing brute force simulation of RVQ, which becomes
infeasible for $B$ larger than 15 or 20, the statistics of RVQ
can be exploited to efficiently and exactly emulate the quantization process:
\begin{enumerate}
\item Draw a realization of the quantization error $Z$ according to its
known CDF (Lemma \ref{lemma-quant_error}).
\item Draw a realization of the corresponding quantization vector according to:
\begin{eqnarray*}
\hat{\bh_k} = \left(\sqrt{1 - Z}\right) {\bf u} + \sqrt{Z} {\bf s}
\end{eqnarray*}
where ${\bf u}$ is isotropic in span($\bH_k$), ${\bf s}$ is isotropic
in the nullspace of span($\bH_k$), with ${\bf u}$, ${\bf s}$ independent.

\end{enumerate}
These steps exactly emulate step 2 of QBC.  The same procedure can also be used to emulate
antenna selection, quantization of the maximum eigenvector, and no combining ($N=1$).
Because the CDF of the quantization error is not known for MRC,
MRC results are generated using brute force RVQ.

\section{Rate Gap with Receiver Estimation Error} \label{sec-rategap_rx_error}

We bound the rate gap using the technique of \cite{Caire_Jindal_ISIT07}.
We first restate the result of Theorem \ref{thm-gap} in terms of
the interference terms $E \left[|(\bheff_k)^H \bv_j |^2 \right]$:
\begin{eqnarray} \label{eq-ratebound2}
\Delta R &\leq& \log_2 e \left( \sum_{l=M-N+1}^{M-1} \frac{1}{l}
\right)+ \log_2 \left(1 +  P \frac{M-1}{M} E \left[
|(\bheff_k)^H \bv_j |^2
\right] \right).
\end{eqnarray}
Using the representation of the channel matrix given in (\ref{eq-HRX}),
we can write the interference term as:
\begin{eqnarray*}
(\bheff_k)^H \bv_j &=& \left (\bH_k {\boldsymbol{\gamma}}_k \right)^H \bv_j
=  \left( \hat{\bG_k}{\boldsymbol{\gamma}}_k \right)^H \bv_j
+ \left( {\bf e}_k {\boldsymbol{\gamma}}_k \right)^H \bv_j.
\end{eqnarray*}
The first term in the sum is statistically identical to the interference term
when there is perfect CSIR, while the second term represents the additional
interference due to the receiver estimation error. Because the noise and
the channel estimate are each zero-mean and are independent we have:
\begin{eqnarray*}
E \left[ |(\bheff_k)^H \bv_j |^2 \right] &=&
E \left[ \left| \left( \hat{\bG_k}{\boldsymbol{\gamma}}_k \right)^H \bv_j \right|^2 \right] +
E \left[ \left| \left( {\bf e}_k {\boldsymbol{\gamma}}_k \right)^H \bv_j \right|^2 \right]
\end{eqnarray*}
The first term comes from the perfect CSIR analysis and is equal to
the product of $\frac{1}{M-1}$ and the expected quantization error with
perfect CSIR.  Because $\boldsymbol{\gamma}_k$ and  $\bv_j$ are each
unit norm and ${\bf e}_k$ is independent of these two vectors, the
quantity $\left( {\bf e}_k {\boldsymbol{\gamma}}_k \right)^H \bv_j$
is (zero-mean) complex Gaussian with variance $(1 + \beta P)^{-1}$,
which is less than $(1 + \beta P)^{-1}$. We finally reach
(\ref{eq-gap_approx2}) by using the approximation for quantization
error from (\ref{eq-approx_quant}) and plugging into
(\ref{eq-ratebound2}), and noting that $(1 + \beta P)^{-1} \approx
(\beta P)^{-1}$.


\end{document}